\documentclass[12pt, a4paper]{article}

\usepackage{latexsym,amsmath,enumerate,amssymb,amsbsy,amsthm,lscape, hyperref} 
\usepackage {graphicx}
\usepackage{float,color,tikz}
\usepackage{subcaption}

\theoremstyle{plain}
\newtheorem{theorem}{Theorem}[section]
\newtheorem{proposition}[theorem]{Proposition}
\newtheorem{lemma}[theorem]{Lemma}
\newtheorem{corollary}[theorem]{Corollary}

\theoremstyle{definition}

\theoremstyle{remark}

\renewcommand\footnotemark{}
\begin{document}

\title
{Characterizations and Constructions of Linear  \\ Intersection Pairs of Cyclic Codes\\ over Finite Fields}

\date{}

\author{Somphong Jitman}

 \thanks{S. Jitman (Corresponding Author)  is  with the Department of Mathematics, Faculty of Science,
 	          Silpakorn University  Nakhon Pathom 73000,  Thailand (sjitman@gmail.com).}

 \thanks{The author declares that there is no conflict of interest regarding the publication of this paper.}

 \thanks{No underlying data was collected or produced in this study.}

 \maketitle

%\subjclass{Primary 11T55,  11T06}
%\keywords{polynimials, roots of polynomials, punctured polynomials,  finite fields}

\begin{abstract} 
 Linear intersection pairs of linear codes have  become of interest due to their nice algebraic properties and wide applications. 
In this paper,  we focus on linear intersection pairs of  cyclic  codes over finite fields.  Some properties of cyclotomic cosets in cyclic groups are presented as key tools in the study of   such linear intersection pairs.  Characterization and constructions of  two cyclic  codes  of a fixed intersecting dimension  are given in terms of their generator polynomials and  cyclotomic cosets. In some cases,  constructions of  two cyclic  codes  of a fixed intersecting subcode are presented as well. 
Based on the theoretical characterization,  some numerical  examples of linear intersection pairs of  cyclic  codes with  good parameters are illustrated.


\smallskip

\end{abstract}

 		\section{Introduction}

 	Cyclic codes form an important family of linear codes that have been extensively studied for both theoretical and practical reasons.    Algebraic structures and properties of cyclic codes have been studied though a one-to-one correspondence between the cyclic codes of length  $n$  over a finite field $\mathbb{F}_q$  and the ideals in the principal ideal ring   $\mathbb{F}_q[x]/\langle x^n-1\rangle $. In applications, cyclic codes are inserting since they can be easily implemented in shift registers \cite{Mac}.  Cyclic codes with additional properties such as self-orthogonality, complementary duality, and self-duality are useful in  practice. In \cite{Huff03} and \cite{KR2012}, characterization and algebraic structures of self-orthogonal cyclic codes have been presented.  Complementary dual cyclic codes over finite fields have  studied in \cite{CG2016} and  \cite{YM1998}. In  \cite{Jia11}, the characterization and enumeration of self-dual cyclic codes have been established.    Based on the algebraic structures of cyclic codes described in \cite{Huff03, Jia11}, hulls and hull dimension of cyclic codes have been studied in \cite{ SJU2015, S2003}.   In \cite{carlet1},   linear complementary pairs (LCPs) of codes have been introduced  as a generalization of  complementary dual codes and  they have been  extensively studied due to their rich algebraic structure and wide applications in cryptography.
 	For example, in \cite{carlet1} and \cite{carlet2},  it was shown that these pairs of codes can be used to counter passive and active side-channel analysis attacks on embedded crypto systems.
 	Several construction of LCPs of codes were also given.

 	In \cite{GGJT2020}, the  author  introduced the notion of a linear intersection pair of linear codes  as well as applications in constructions of  entanglement-assisted quantum error-correcting codes (EAQECCs). Precisely, for a non-negative integer  $\ell$, linear codes  $C$ and $D$  of the same length  $n$  over $\mathbb{F}_q$   form a linear  $\ell$-intersection pair if   $\dim(C\cap D)=\ell$. This concept can be viewed as a generalization of self-orthogonal codes,  complementary dual codes \cite{Jin2017,M1992,QZ2018}, hulls \cite{TJ2020}, and linear complementary pairs \cite{carlet1,carlet2}.  Precisely, a self-orthogonal code   $C$ is a linear  $\ell$-intersection pair  with $D=C^\perp$  and  $\dim(C)=\ell$, a complementary dual code  $C$  is a linear  $\ell$-intersection pair with   $D=C^\perp$  and  $\ell=0$,   $C$ has hull dimension  $\ell$  if  $D=C^\perp$, and  $C$   and   $D$  form a LCP if  $C$ and   $D$ form a linear  $\ell$-intersection pair and  $\dim(C)+\dim(D)=n$.

 	In this paper,         we focus on linear intersection pairs of  cyclic  codes over finite fields.  Some properties of cyclotomic cosets in cyclic groups are presented as key tools in the study of   such linear intersection pairs.  Characterization and constructions of  cyclic  codes pairs  of a fixed intersecting dimension  are given in terms of their generator polynomials and  cyclotomic cosets. In some cases,  constructions of  two cyclic  codes   of a fixed intersecting subcode are presented as well. 
 	Based on the theoretical characterization,  some numerical  examples of linear intersection pairs of  cyclic  codes with  good parameters.

 	The paper is organized as follows. Some preliminary results are recalled in Section 2. In Section 3, properties of cyclotomic cosets of cyclic codes are presented as well as characterization and constructions of  two cyclic   codes  of a fixed intersecting dimension. 
 	Numerical  examples of linear intersection pairs of  cyclic  codes with  good parameters are presented in Section 4. 
 	Discussion and summary are  given in Section 5.

 	\section{Preliminary}
 	
 	In this section, basic concepts and properties of linear codes, cyclotomic cosets, cyclic codes and linear intersection pairs are recalled. The reader may refer to  \cite{GGJT2020}, \cite{Huff03},  \cite{JLLX2012}, and \cite{Lidl97} for more details. 
 	
 	\subsection{Linear Codes and   Linear Intersection Pairs}
 	
 	For a prime power $q$, let $\mathbb{F}_q$ denote the  finite field of $q$ elements. For a positive integer $n$, let $\mathbb{F}_q^n$ denote the vector space of all $n$-tuples over $\mathbb{F}_q$. A \textit{linear code} $C$ of length $n$ and dimension $k$ over $\mathbb{F}_q$ is defined to be a $k$-dimensional subspace of $\mathbb{F}_q^n$.   In this case, $C$ is called an   $[n,k]_q$  code. In addition,   if ${\rm wt}(C)=\min\{ {\rm wt}(\boldsymbol{c})\mid \boldsymbol{c} \in C\setminus \{\boldsymbol{0}\}\}$ is $d$,  the code $C$ is referred  to as an  $[n,k,d]_q$  code, where  $ {\rm wt}(\boldsymbol{x})=|\{i\mid x_i\ne 0\}$ for all $ {\rm wt}(\boldsymbol{x})=(x_1,x_2,\dots, x_n)\in \mathbb{F}_q^n$.
 	The \textit{dual} of a linear code $C$ of length $n$ over $\mathbb{F}_q$  is defined to be a linear code of the form 
 	\[C^\perp=\left\{(x_0,\ldots,x_{n-1})\in\mathbb{F}_q^n \,\left\vert\, \sum_{i=0}^{n-1} x_ic_i=0 ~\text{ for all}~ (c_0,\ldots,c_{n-1})\in C \right.\right\}\]
 	and the \textit{hull} of $C$  is defined to be $Hull(C)=C\cap C^\perp$.  A linear code  $C$  is said to be \textit{self-orthogonal} if   $C\subseteq C^\perp$. A linear code  $C$  is said to be \textit{self-dual} 
 	(resp., \textit{linear complementary dual} (LCD))  if    $C=Hull(C)=C^\perp$   (resp., $Hull(C)=\{0\}$).  
 	
 	Linear codes  $C$ and $D$  of the same length  $n$  over $\mathbb{F}_q$   are said to be a {\em linear complementary pair} (LCP)  \cite{carlet2} if $C\cap D=\{0\}$ and $C+ D=\mathbb{F}_q^n$. In \cite{GGJT2020},  linear intersection pair of linear codes has been introduced as a generalization of hulls and LCPs.  Precisely, for a non-negative integer  $\ell$, linear codes  $C$ and $D$  of the same length  $n$  over $\mathbb{F}_q$   form a {\em linear  $\ell$-intersection pair} if   $\dim(C\cap D)=\ell$.  Without loss of generality, it can be assumed that $0\leq \ell \leq \dim(C)\leq \dim(D) \leq n$.  
 	
 	The concept of linear intersection pairs can be viewed as a generalization of self-orthogonal codes, complementary dual codes, hulls, and linear complementary pairs.  
 	Precisely,  a self-orthogonal code $C$ induces a linear $\ell$-intersection pair  with $D=C^\perp$  and $\ell=\dim(C)$, a complementary dual code $C$ induces a linear $\ell$-intersection pair with
 	$D=C^\perp$ and $\ell=0$,  $C$ has hull dimension $\ell$ if $D=C^\perp$, and $C$ and $D$ form an LCP if $C$ and $D$ form a linear $\ell$-intersection pair such 
 	that $\dim(C)+\dim(D)=n$.
 	Preliminary results on linear $\ell$-intersection pairs of linear codes  have been investigated in \cite{GGJT2020}.

 	\subsection{Cyclic Codes}
 	
 	A linear code $C$ of length $n$ over $\mathbb{F}_q$ is said to be \textit{cyclic}  if 
 	\begin{center}
 		$(c_{n-1}, c_0,\ldots, c_{n-2})\in C$ whenever $(c_0, c_1,\ldots, c_{n-1})\in C$.
 	\end{center}
 	It is well known  (see, e.g., \cite{Huff03})) that every  cyclic code  $C$  of length $n$ over $\mathbb{F}_q$ can be viewed as an ideal in the principal ideal  ring $\mathbb{F}_q[x]/\left\langle x^n-1\right\rangle$ generated by a unique monic divisor of $x^n-1$ of minimal degree in the said ideal. Such a monic polynomial is called the \textit{generator polynomial} of the cyclic code $C$.   
 	
 	\begin{theorem}\label{genC}
 		Let $C$  be a cyclic code of length $n$ over $\mathbb{F}_q$ generated by $g(x)$ and let $k$ be an integer such that $0\leq k\leq n$. Then $C$  has dimension $k$ if and only if $\deg(g(x))=n-k$.
 	\end{theorem}

 	The following theorem  (cf.  \cite[Theorem 4.3.7]{Huff03}) is useful in the  investigation of  algebraic structures of  a linear $\ell$-intersection pair of  cyclic codes.

 	\begin{theorem}\label{lcm}
 		Let $C_1$ and $C_2$ be cyclic codes of length $n$ over $\mathbb{F}_q$ generated by $g_1(x)$ and $g_2(x)$, respectively. Then $C_1\cap C_2$ is a cyclic code generated by 
 		${\rm lcm}(g_1(x), g_2(x))$ and  $C_1+C_2$ is a cyclic code generated by 
 		${\rm gcd}(g_1(x), g_2(x))$.
 	\end{theorem}

 	\subsection{Cyclotomic Cosets  in  Finite Cyclic Groups}

 	As monic  divisors of  $x^n-1\in \mathbb{F}_q[x]$ are key in the study of cyclic codes, we recall the factorization of  $x^n-1$   into a  product of monic irreducible polynomials.  Throughout the rest of this paper,  we assume  that $\mathbb{F}_q$ is a finite field of characteristic $p$ and order $q$.  Given a positive integer $n$, we can write $n=p^{\nu}{n'}$,  where $p\nmid n'$ and $\nu$ is an integer.  In $\mathbb{F}_q[x]$, it is not difficult to see that $x^n-1= (x^{n'}-1)^{p^\nu}$,  and hence,  it is sufficient to focus on the factorization of  $x^{n'}-1$.

 	For coprime positive integers $i$ and $j$,  let $ {\rm ord}_{j}(i)$ denote the multiplicative order of $i$ modulo $j$.  Let $q$ be a prime power and let $n' $ be a positive integer   such that  $\gcd(n',q)=1$.   For  each $a\in \mathbb{Z}_{n'}$, denote by ${\rm ord}(a)$ the additive order of $a$ in $\mathbb{Z}_{n'}$. 
 	For  each $a\in \mathbb{Z}_{n'}$,  the {\em $q$-cyclotomic  coset  of $\mathbb{Z}_{n'}$ containing $a$} is defined to be  the set 	\begin{align*}
 		S_q(a):=&\{q^i\cdot a \mid i=0,1,\dots\}, 
 	\end{align*}
 	where $q^i\cdot a:= \sum\limits_{j=1}^{q^i}a$ in $\mathbb{Z}_{n'}$.

 	From \cite{JLLX2012}, we have the following results. 
 	
 	\begin{proposition} \label{prop:coset} Let $q$ be a prime power  such that  $\gcd(n',q)=1$ and let  $a\in \mathbb{Z}_{n'}$.  Then the following statements hold. 
 		
 		\begin{enumerate}
 			\item $\gcd({\rm ord}_{{\rm ord}(a)}(q) ,q)=1$ and $	S_q(a)=\{q^i\cdot a  \mid 0\leq i< {\rm ord}_{{\rm ord}(a)}(q) \}$. Equivalently,  $|S_q(a)|={\rm ord}_{{\rm ord}(a)}(q) $.
 			\item The elements  in  $S_q(a) $ have the same order. 
 			\item The elements in $ \mathbb{Z}_{n'}$ of  the same order are partitioned in to $q$-cyclotomic cosets of the same size. 
 			\item   The cyclic group $ \mathbb{Z}_{n'} $ can be decomposed of the form 
 			\begin{align}\mathbb{Z}_{n'} =\bigcup_{d|n'} A_d= \bigcup_{d|n'} \bigcup_{a\in A_d} S_q(a) , \label{decomZn}\end{align}
 			where $A_d$ is the set of elements of order $d$ in $ \mathbb{Z}_{n'}$.  The union can be re-indexed  to be  disjoint.
 			\item The number of $q$-cyclotomic cosets in $ \mathbb{Z}_{n'}$ is
 			
 			\[   \sum_{d|n'} \dfrac{\phi(d)}{{\rm ord}_{ d}(q)},\] 
 		\end{enumerate}
 		where $\phi$ is the  Euler's totient function. 
 	\end{proposition}

 	From  \cite{Lidl97} and  \eqref{decomZn},  there is a one-to-one correspondence between the irreducible factors of  $x^{{n'}}-1$ and the $q$-cyclotomic cosets of $\mathbb{Z}_{n'}$  and we have

 	\begin{align}  x^{{n'}}-1=  \prod_{d|n'} \prod_{i=1}^{\frac{\phi(d)}{{\rm ord}_{ d}(q)}} f_{d,i}(x), \label{xn-1}\end{align}
 	where  $f_{d,i}(x)$ is an irreducible polynomial determined by some $q$-cyclotomic coset of  $\mathbb{Z}_{n'}$.  Precisely,  \[f_{d,i}(x)= \prod_{j\in S_q(a)} (x-\alpha^j)\]  for some  $a\in  \mathbb{Z}_{n'}$, where $\alpha $ is a primitive $n'$th root of unity.  In this case,  $\deg(f_{d,i}(x))= |S_q(a)|= {\rm ord}_{ {\rm ord}(a)}(q)$.

 	\section{Linear Intersection Pairs of Cyclic Codes}
 	In this section,  characterization on the existence of linear intersection pairs of cyclic codes is given  together with their constructions.

 	From  Theorem \ref{genC} and Theorem \ref{lcm}, the following useful theorem can be deduced directly. 
 	\begin{theorem}  \label{thm:degGen}
 		Let $C_1$ and $C_2$ be cyclic codes of length $n$ over $\mathbb{F}_q$ generated by $g_1(x)$ and $g_2(x)$, respectively. Let $\ell$ be an integer such that $0\leq \ell \leq n$. Then $C_1$ and $ C_2$ form  a linear $\ell$-intersection pair  if and only if 
 		$\deg({\rm lcm}(g_1(x), g_2(x)))=n-\ell$.  
 	\end{theorem}
 	
 	\subsection{The Existence of Linear Intersection Pairs of Cyclic Codes}
 	
 	Fist, we  give necessary and sufficient conditions to have  a linear  $\ell$-intersection  pair of cyclic codes of length $n$ over $\mathbb{F}_q$ regardless the dimensions of the two cyclic codes. This is equivalent to the existence of a cyclic code of length $n$ and dimension $\ell$ over  $\mathbb{F}_q$ which is equivalent to that of its dual has dimension $n-\ell$. Precisely, each   $\ell$-dimensional  cyclic code of length $n$ over $\mathbb{F}_q$ can be extended to  have a pair of cyclic codes which is a linear  $\ell$-intersection  pair.  This can be  summarized  in the following lemma.

 	\begin{lemma} \label{lem:char-non-dim}
 		Let $\mathbb{F}_q$ be a finite field   and let $n$ be a positive integer. Let $\ell$  be  an integer such  that   $0\leq  \ell  \leq n$.   Then  the following statements are equivalent.

 		\begin{enumerate}[$i)$]
 			\item There exists  a linear $\ell$-intersection pair of  cyclic codes of length $n$ over $\mathbb{F}_q$.
 			\item  $x^n-1$ has a monic divisor of degree $n-\ell$.
 			\item  $x^n-1$ has a monic divisor of degree $\ell$. 
 		\end{enumerate}
 	\end{lemma}
 	
 	More precise and practical conditions for the  existence of  a linear $\ell$-intersection pair of  cyclic codes of length $n$ over $\mathbb{F}_q$  are given in the following proposition and corollary.
 	
 	\begin{proposition} \label{prop:exist} Let $\mathbb{F}_q$ be a finite field of characteristic $p$ and let $n=p^\nu n'$ be a positive integer such that   $p\nmid n'$ and $\nu\geq 0$. Let $\ell$  be  an integer such  that   $0\leq  \ell  \leq n$.   Then there exists  a linear $\ell$-intersection pair of  cyclic codes of length $n$ over $\mathbb{F}_q$ if and only if 
 		
 		\[ \sum _{d|n'}   {\rm ord}_{ d}(q)\cdot s_d =\ell\] 
 		has an integer  solution such that $0\leq s_d \leq   \dfrac{p^\nu \phi(d)}{{\rm ord}_{ d}(q)}$ for all $d|n'$.
 	\end{proposition}
 	\begin{proof}  From \eqref{xn-1}, it can be deduced that 
 		\begin{align} \label{eq:xn-1} x^n-1=\left(x^{{n'}}-1\right)^{p^\nu}= \prod_{d|n'} \prod_{i=1}^{\frac{\phi(d)}{{\rm ord}_{ d}(q)}} (f_{d,i}(x))^{p^\nu},\end{align}
 		where  $f_{d,i}(x)$ is an irreducible polynomial of degree $ {\rm ord}_{ d}(q)$ induced by a $q$-cyclotomic coset in $\mathbb{Z}_{n'}$ of size  $ {\rm ord}_{ d}(q)$.   
 		
 		Assume that  there exists  a linear $\ell$-intersection pair of  cyclic codes of length $n$ over $\mathbb{F}_q$.  By Lemma \ref{lem:char-non-dim},  $x^n-1$ has a monic divisor of degree $\ell$.   Then 
 		\[ \ell= \deg\left(  \prod_{d|n'} \prod_{i=1}^{\frac{\phi(d)}{{\rm ord}_{ d}(q)}} (f_{d,i}(x))^{j_{d,i} }\right)  = \sum_{d|n'} \sum_{i=1}^{\frac{\phi(d)}{{\rm ord}_{ d}(q)}}  \deg\left( (f_{d,i}(x))^{j_{d,i} }\right)  =    \sum_{d|n'}  \left( {\rm ord}_{ d}(q) \sum_{i=1}^{\frac{\phi(d)}{{\rm ord}_{ d}(q)}}    {j_{d,i} } \right)  \]
 		for some  $0\leq {j_{d,i}} \leq p^\nu$.   
 		The result follows.

 		Conversely,  assume that 	\[ \sum _{d|n'}   {\rm ord}_{ d}(q)\cdot s_d =\ell\] 
 		has an integer  solution such that $0\leq s_d \leq   \dfrac{p^\nu \phi(d)}{{\rm ord}_{ d}(q)}$ for all $d|n'$. Based on \eqref{eq:xn-1},  a monic  divisor  $g(x)$  of $x^n-1$  can be chosen to have  degree $\ell$. Hence,  there exists  a linear $\ell$-intersection pair of  cyclic codes of length $n$ over $\mathbb{F}_q$  by Lemma \ref{lem:char-non-dim}.
 	\end{proof}
 	
 	The next corollary can be deduced directly form  Proposition \ref{prop:coset} and Proposition \ref{prop:exist}.
 	\begin{corollary}
 		Let $\mathbb{F}_q$ be a finite field of characteristic $p$ and let $ n'$ be a positive integer such that   $p\nmid n'$. Let $\ell$  be  an integer such  that   $0\leq  \ell  \leq n$.   Then the following statements are equivalent.
 		\begin{enumerate}[$i)$]
 			\item  There exists  a linear $\ell$-intersection pair of  cyclic codes of length $n'$ over $\mathbb{F}_q$.
 			\item The equation
 			
 			\[ \sum _{d|n'}   {\rm ord}_{ d}(q)\cdot x_d =\ell\] 
 			has an integer  solution such that $0\leq x_d \leq   \dfrac{ \phi(d)}{{\rm ord}_{ d}(q)}$ for all $d|n'$.
 			\item  There exist $q$-cyclotomic cosets of   $\mathbb{Z}_{n'}$ whose union has cardinality $\ell$.
 		\end{enumerate}
 	\end{corollary}

 	For a small value of $\ell$, we have the following results. 
 	\begin{theorem}\label{thm:charExs} Let $\mathbb{F}_q$ be a finite field of characteristic $p$ and let $n=p^\nu n'$ be a positive integer such that   $p\nmid n'$ and $\nu\geq 0$. Let $\ell$, $k_1$, and $k_2$ be  integers such  that   $0\leq \ell \leq  k_1\leq k_2 \leq n$.  Then the following statements holds.
 		\begin{enumerate}[$i)$]
 			\item  If there exists a linear $\ell$-intersection pair of $[n,k_1]$ and $[n,k_2]_q$ cyclic codes, then   $k_1+k_2 -\ell \leq n$.
 			
 			\item  There always exists a linear $0$-intersection pair of cyclic codes of length $n$ over $\mathbb{F}_q$.
 			
 			\item  There always exists a linear $1$-intersection pair of cyclic codes of length $n$ over $\mathbb{F}_q$.
 			
 			\item    If $n$ is even, $\nu\geq 1$,  or  ${\rm ord}_{ d}(q)=2$ for some $d|n'$, there exists   a linear $2$-intersection pair of cyclic codes of length $n$ over $\mathbb{F}_q$.
 		\end{enumerate}
 	\end{theorem}
 	\begin{proof}
 		To prove $i)$,  assume that  there exists a linear $\ell$-intersection pair of $[n,k_1]$ and $[n,k_2]_q$ cyclic codes with generator polynomials $g_1(x)$ and $g_2(x)$, respectively.  Then \[n-\ell =\deg({\rm  lcm} (g_1(x),g_2(x))\leq \deg(g_1(x))+\deg(g_2(x)) =2n-k_1-k_2\] which implies that  $k_1+k_2 -\ell \leq n$.
 		
 		Since $1$ is a monic divisor of $x^n-1$ of degree $0$,  there  exists a linear $0$-intersection pair of cyclic codes of length $n$ over $\mathbb{F}_q$ by Proposition \ref{prop:exist}. Hence, $ii)$ is proved.

 		Since $x-1$ is a monic divisor of $x^n-1$ of degree $1$,  there  exists a linear $1$-intersection pair of cyclic codes of length $n$ over $\mathbb{F}_q$ by Proposition \ref{prop:exist}. The statement $iii)$ follows immediately.

 		To prove $iv)$, assume that  $n$ is even, $\nu\geq 1$,  or  ${\rm ord}_{ d}(q)=2$ for some $d|n'$. 
 		
 		\noindent {\bf Case 1:} $n$ is even.  We have that $x^2-1$  is a monic divisor of $x^n-1$ of degree $2$. 
 		
 		\noindent {\bf Case 2:} $\nu\geq 1$.  It follows that  $(x-1)^2$  is a monic divisor of $x^n-1$ of degree $2$. 
 		
 		\noindent {\bf Case 3:}  ${\rm ord}_{ d}(q)=2$ for some $d|n'$.  Then there exists a $q$-cyclotomoc coset in $\mathbb{Z}_{n'}$  of size  ${\rm ord}_{ d}(q)=2$.  Hence,   $x^{n'}-1$ (resp., $x^{n}-1$) has a monic divisor of degree $2$.

 		From the 3 cases,  there exists a monic divisor of $x^n-1$ of degree $2$. Hence,  there  exists a linear $2$-intersection pair of cyclic codes of length $n$ over $\mathbb{F}_q$ by Proposition \ref{prop:exist}.
 	\end{proof}

 	\subsection{Constructions of Linear Intersection Pairs of Cyclic Codes}
 	
 	Based on results in the previous section,  constructions of some linear intersection pairs of cyclic codes with prescribe intersecting dimension are derived. 
 	
 	\begin{theorem} \label{thm:L(x)} Let $\mathbb{F}_q$ be a finite field of characteristic $p$ and  let  $n$  be a positive integer. Let $L(x)$ be a monic divisor of $x^n-1$ of degree $\ell$.  Let $g_1(x)$ and $g_2(x)$ be  monic  polynomials such that   $g_2(x)|g_1(x)|\dfrac{x^n-1}{L(x)}$ over  $\mathbb{F}_q$. 
 		Then the cyclic codes with generator polynomials $g_1(x)$  
 		and 
 		$k(x):=\dfrac{ g_2(x)(x^n-1)}{L(x)g_1(x)}$ form a linear $s$-intersection pair for some $\ell\leq s \leq \ell+\deg(g_1(x))$.
 		
 		In particular, if   $\gcd\left(g_1(x),  \dfrac{x^n-1}{L(x)g_1(x)}\right)=1$, then  the cyclic codes with generator polynomials $g_1(x)$  
 		and 
 		$k(x) $ form a linear $\ell$-intersection pair.
 	\end{theorem}
 	\begin{proof} It is not difficult to see that $k(x)$ is a divisor of $x^n-1$. Let $C_1$ and $C_2$ be the cyclic codes  with generator polynomials $g_1(x)$ and $k(x) $, receptively.  Clearly, \[ \dfrac{x^n-1}{L(x)g_1(x)} | {\rm lcm} \left( g_1(x),\dfrac{x^n-1}{L(x)g_1(x)} \right) | \dfrac{x^n-1}{L(x)}.\] Since $  g_2(x)|g_1(x)$, it follows that   \[{\rm lcm} ( g_1(x),k(x))= {\rm  lcm} \left( g_1(x),  \dfrac{g_2(x)(x^n-1)}{L(x)g_1(x) }\right) ={\rm lcm} \left( g_1(x),\dfrac{x^n-1}{L(x)g_1(x)} \right).\]
 		Hence,  $ n-(\ell+\deg(g_1(x)))\leq \deg ({\rm lcm} ( g_1(x),k(x))\leq n-\ell$. 
 		From Theorem \ref{thm:degGen},  $C_1$ and $C_2 $ form a  linear $s$-intersection pair of cyclic codes  for some $\ell\leq s \leq \ell+\deg(g_1(x))$.
 		
 		Next, assume that   $\gcd\left(g_1(x),  \dfrac{x^n-1}{L(x)g_1(x)}\right)=1$.  Then  \[{\rm lcm} ( g_1(x),k(x)) ={\rm lcm} \left( g_1(x),\dfrac{x^n-1}{L(x)g_1(x)} \right)  = \dfrac{x^n-1}{L(x)}\] which implies that 
 		$\deg ({\rm lcm} ( g_1(x),k(x))= n-\ell$. 
 		From Theorem \ref{thm:degGen},  $C_1$ and $C_2 $ form a  linear $\ell$-intersection pair.
 	\end{proof}
 	
 	For simple root cyclic codes, we have the following corollary. 
 	\begin{corollary} \label{cor:L(x)}
 		Let $\mathbb{F}_q$ be a finite field and let   $n$  be a positive integer such that $\gcd(n,q)=1$.  Let $L(x)$ be a monic divisor of $x^n-1$ of degree $\ell$.  Let $g_1(x)$ and $g_2(x)$ be  monic  polynomials such that   $g_2(x)|g_1(x)|\dfrac{x^n-1}{L(x)}$ over  $\mathbb{F}_q$. 
 		Then the cyclic codes with generator polynomials $g_1(x)$  
 		and 
 		$k(x):=\dfrac{ g_2(x)(x^n-1)}{L(x)g_1(x)}$ form a linear   $\ell$-intersection pair
 	\end{corollary}
 	\begin{proof}
 		Since  $\gcd(n,q)=1$, it follows that  $\gcd\left(g_1(x),  \dfrac{x^n-1}{L(x)g_1(x)}\right)=1$.  The result follows from Theorem \ref{thm:L(x)}.
 	\end{proof}
 	
 	\begin{corollary} \label{cor:Lex-main} Let $\mathbb{F}_q$ be a finite field of characteristic $p$ and let $n=p^\nu n'$ be a positive integer such that   $p\nmid n'$ and $\nu\geq 0$.   Let $L(x)$ be a monic divisor of $x^{n'}-1$ of degree $\ell$.  Let  $g_1(x)$ and $g_2(x)$ be  monic  polynomials such that   $g_1(x)|\dfrac{x^{n'}-1}{L(x)}$  and  $g_2(x)|(g_1(x))^{p^\nu}$ over  $\mathbb{F}_q$.  
 		Then the cyclic codes with generator polynomials $(g_1(x))^{p^\nu}$  
 		and 
 		$k(x):=\dfrac{ g_2(x)(x^n-1)}{(L(x))^s(g_1(x))^{p^\nu}}$ form a linear   $\ell s$-intersection pair for all integers $0\leq s \leq p^\nu$.
 		
 	\end{corollary}
 	\begin{proof}   Let $s$ be a integer such that  $0\leq s \leq p^\nu$. It is not difficult to see that  $g_2(x)|(g_1(x))^{p^\nu}| \dfrac{x^n-1}{(L(x))^s}$ and  $ \gcd\left((g_1(x))^{p^\nu},  \dfrac{x^n-1}{(L(x))^s(g_1(x))^{p^\nu}}\right)=1$. The result therefore  follows from  Theorem  \ref{thm:L(x)}. 
 	\end{proof}

 	\begin{corollary} Let $\mathbb{F}_q$ be a finite field of characteristic $p$ and let $n=p^\nu n'$ be a positive integer such that   $p\nmid n'$ and $\nu\geq 0$.  Let $g_1(x)$ and $g_2(x)$ be  monic  polynomials such that   $g_1(x)|x^{n'}-1$  and  $g_2(x)|(g_1(x))^{p^\nu}$ over  $\mathbb{F}_q$.  
 		Then the cyclic codes with generator polynomials $(g_1(x))^{p^\nu}$  
 		and 
 		$k(x):=\dfrac{ g_2(x)(x^n-1)}{(g_1(x))^{p^\nu}}$ form a linear   $0$-intersection pair.
 	\end{corollary}
 	\begin{proof}  By setting $L(x)=1$ and $s=1$, the result follows from     Corollary \ref{cor:Lex-main}. 
 	\end{proof}

 	\begin{corollary} Let $\mathbb{F}_q$ be a finite field of characteristic $p$ and  let $n=p^\nu n'$ be a positive integer such that   $p\nmid n'$ and $\nu\geq 0$.    Let $g_1(x)$ and $g_2(x)$ be  monic  polynomials such that   $g_1(x)|\dfrac{x^{n'}-1}{x-1}$  and  $g_2(x)|(g_1(x))^{p^\nu}$ over  $\mathbb{F}_q$.  
 		Then the cyclic codes with generator polynomials $ (g_1(x))^{p^\nu}$ and $k(x):=\dfrac{ g_2(x)(x^n-1)}{(x-1)^s(g_1(x))^{p^\nu}}$ form a linear $s$-intersection pair for all integers $0\leq s \leq p^\nu$.
 	\end{corollary}
 	\begin{proof}
 		By setting $L(x)=x-1$, the result follows  from  Corollary \ref{cor:Lex-main}.
 	\end{proof}
 	
 	\begin{corollary} Let $\mathbb{F}_q$ be a finite field of characteristic $p$ and let $n=p^\nu n'$ be a positive integer such that   $p\nmid n'$ and $\nu\geq 0$.   Then the following statements hold.
 		
 		\begin{enumerate}[$(i)$]
 			\item If $n'$ is even,   then  $x^2-1$ divides $x^{n'}-1$. In this case, if  $g_1(x)$ and $g_2(x)$ are  monic  polynomial  such that  $g_1(x)|\dfrac{x^{n'}-1}{x^2-1}$ over  $\mathbb{F}_q$ and $g_2(x)|(g_1(x))^{p^\nu}$, then 
 			the cyclic codes with generator polynomials $(g_1(x))^{p^\nu}$ and $\dfrac{g_2(x)(x^n-1)}{(x^2-1)^s (g_1(x))^{p^\nu}}$ form a linear $2s$-intersection pair for  all integers $0\leq s \leq p^\nu$. 
 			\item If  ${\rm ord}_{ d}(q)=2$ for some $d|n'$, then    $x^{n'}-1$ has a monic irreducible factor $a(x)$ of degree $2$.  In this case,   if  $g_1(x)$ and $g_2(x)$ are  monic  polynomial  such that  $g_1(x)|\dfrac{x^{n'}-1}{a(x)}$ over  $\mathbb{F}_q$ and $g_2(x)|(g_1(x))^{p^\nu}$,  then the cyclic codes with generator polynomials $(g_1(x))^{p^\nu}$ and $\dfrac{g_2(x)(x^n-1)}{(a(x))^s (g_1(x))^{p^\nu}}$ form a linear $2s$-intersection pair all integers $0\leq s \leq p^\nu$. 
 		\end{enumerate}
 	\end{corollary}
 	\begin{proof}
 		By setting $L(x)=x^2-1$ (resp., $L(x)=a(x)$), the result follows  from  Corollary \ref{cor:Lex-main}.
 	\end{proof}
 	
 	From the corollaries above, various   linear $\ell$-intersection pairs  of cyclic codes can be constructed by the variation of   $g_1(x)$ and $g_2(x) $.    For $\ell\in \{0,1,2\}$,  some illustrative  examples of  linear $\ell$-intersection pairs  of cyclic codes of length $n$ over $\mathbb{F}_2$  are given in Section~\ref{sec:Numer}
 	
 	\subsection{MDS Linear Intersection Pairs of Cyclic Codes}
 	
 	A linear $[n,k,d]_q$ code is said to be  Maximum Distance Separable (MDS)  if  the parameters attain   the Singleton bound $d\leq n-k+1$. 
 	In \cite{GGJT2020},    linear  $\ell$-intersection pairs of MDS linear codes  have been  derived via  the family of Generalized Reed-Solomon codes.   Here, we focus on constructions of   linear  $\ell$-intersection pairs of MDS cyclic codes using  Reed-Solomon codes.

 	\begin{theorem} Let $q$  be a prime power and let $n$  be a positive  integer such that $n|(q-1)$.   Let $k_1$ and $k_2$ be integers such that $0\leq k_1\leq k_2\leq n$. Then there exist a linear $\ell$-intersection pair of  $[n,k_1,n-k_1+1] _q$ and $[n,k_2,n-k_2+1] _q$  MDS  cyclic  codes 
 		for all integers $0\leq \ell \leq k_1$ such that  $k_1+k_2 -\ell \leq n$.
 		
 	\end{theorem}
 	\begin{proof}  Let $\ell$ be an integer such that $0\leq \ell \leq k_1$  and  $k_1+k_2 -\ell \leq n$.  Let $\alpha $ be a primitive $n$th root of unity. Let  $C_1$ and $C_2$ be  Reed-Solomon codes with generator polynomials
 		\[g_1(x) =\prod_{i=0}^{n-k_1-1} (x-\alpha^i) \text{ and }   g_2(x) =\prod_{i=k_2-\ell}^{n-\ell-1} (x-\alpha^i) ,\] respectively. 
 		Since $k_2-\ell  \leq n-k_1-1$,   we have that  \[{\rm lcm}(g_1(x),g_2(x))= \prod_{i=0}^{n-\ell-1} (x-\alpha^i)\] is of  degree $n-\ell$. By Theorem \ref{thm:degGen}, it follows that  $C_1$ and $C_2$ form a  linear $\ell$-intersection pair of cyclic codes.  We note that $\deg(g_1(x))=n-k_1$ and $\deg(g_2(x))=n-k_2$. Using the properties of Reed-Solomon codes, $C_1$ and $C_2$ are MDS with parameters $[n,k_1,n-k_1+1] _q$ and $[n,k_2,n-k_2+1] _q$, respectively.
 	\end{proof}

 	\section{Numerical Results} \label{sec:Numer}
 	In this section,  some illustrative examples of linear intersection pairs of cyclic codes over $\mathbb{F}_2$ are given based on theoretical results in Section 3.  The codes presented in Table \ref{t1}--\ref{t3} are optimal.     The calculation in this paper is done using the  computer algebra Magma \cite{BCP1997}. The optimal codes refer to the database \cite{G2023}.

 	\begin{landscape}
 		
 		\begin{table}
 			\centering
 			\begin{tabular}{|c|c|l|l|}
 				\hline
 				$C$ &$D$ & $g_C(x)$& $g_D(x)$  \\ \hline
 				
 				$[ 7 , 3 , 4 ]_ 2 $& $[ 7 , 3 , 4 ]_ 2 $&$ x^4 + x^2 + x + 1
 				$&$ x^4 + x^3 + x^2 + 1
 				$ \\ \hline
 				
 				$[ 7 , 4 , 3 ]_ 2 $& $[ 7 , 3 , 4 ]_ 2 $&$ x^3 + x^2 + 1
 				$&$ x^4 + x^3 + x^2 + 1
 				$ \\ \hline

 				$[ 9 , 2 , 6 ]_ 2 $& $[ 9 , 6 , 2 ]_ 2 $&$ x^7 + x^6 + x^4 + x^3 + x + 1
 				$&$ x^3 + 1
 				$ \\ \hline
 				
 				$[ 9 , 7 , 2 ]_ 2 $& $[ 9 , 2 , 6 ]_ 2 $&$ x^2 + x + 1
 				$&$ x^7 + x^6 + x^4 + x^3 + x + 1
 				$ \\ \hline
 				
 				$[ 15 , 2 , 10 ]_ 2 $& $[ 15 , 4 , 8 ]_ 2 $&$ x^{13} + x^{12} + x^{10} + x^9 + x^7 +
 				x^6 + x^4 + x^3 + x + 1
 				$&$ x^{11} + x^8 + x^7 + x^5 + x^3 + x^2 + x + 1
 				$ \\ \hline						
 				$[ 15 , 2 , 10 ]_ 2 $& $[ 15 , 8 , 4 ]_ 2 $&$ x^{13} + x^{12} + x^{10} + x^9 + x^7 +
 				x^6 + x^4 + x^3 + x + 1
 				$&$ x^7 + x^6 + x^5 + x^2 + x + 1
 				$ \\ \hline

 				$[ 15 , 2 , 10 ]_ 2 $& $[ 15 , 12 , 2 ]_ 2 $&$ x^{13} + x^{12} + x^{10} + x^9 + x^7 +
 				x^6 + x^4 + x^3 + x + 1
 				$&$ x^3 + 1
 				$ \\ \hline

 				$[ 15 , 4 , 8 ]_ 2 $& $[ 15 , 8 , 4 ]_ 2 $&$ x^{11} + x^8 + x^7 + x^5 + x^3 + x^2
 				+ x + 1
 				$&$ x^7 + x^3 + x + 1
 				$ \\ \hline
 				$[ 15 , 4 , 8 ]_ 2 $& $[ 15 , 4 , 8 ]_ 2 $&$ x^{11} + x^8 + x^7 + x^5 + x^3 + x^2
 				+ x + 1
 				$&$ x^{11} + x^{10} + x^9 + x^8 + x^6 + x^4 + x^3 + 1
 				$ \\ \hline
 				
 				$[ 15 , 5 , 7 ]_ 2 $& $[ 15 , 2 , 10 ]_ 2 $&$ x^{10} + x^9 + x^8 + x^6 + x^5 + x^2
 				+ 1
 				$&$ x^{13} + x^{12} + x^{10} + x^9 + x^7 + x^6 + x^4 + x^3 + x + 1
 				$ \\ \hline				
 				$[ 15 , 5 , 7 ]_ 2 $& $[ 15 , 6 , 6 ]_ 2 $&$ x^{10} + x^9 + x^8 + x^6 + x^5 + x^2
 				+ 1
 				$&$ x^9 + x^6 + x^5 + x^4 + x + 1
 				$ \\ \hline

 				$[ 15 , 5 , 7 ]_ 2 $& $[ 15 , 8 , 4 ]_ 2 $&$ x^{10} + x^9 + x^8 + x^6 + x^5 + x^2
 				+ 1
 				$&$ x^7 + x^3 + x + 1
 				$ \\ \hline

 				$[ 15 , 5 , 7 ]_ 2 $& $[ 15 , 10 , 4 ]_ 2 $&$ x^{10} + x^9 + x^8 + x^6 + x^5 + x^2
 				+ 1
 				$&$ x^5 + x^4 + x^2 + 1
 				$ \\ \hline

 				$[ 15 , 6 , 6 ]_ 2 $& $[ 15 , 4 , 8 ]_ 2 $&$ x^9 + x^8 + x^5 + x^4 + x^3 + 1
 				$&$ x^{11} + x^{10} + x^9 + x^8 + x^6 + x^4 + x^3 + 1
 				$ \\ \hline
 				
 				$[ 15 , 6 , 6 ]_ 2 $& $[ 15 , 8 , 4 ]_ 2 $&$ x^9 + x^6 + x^5 + x^4 + x + 1
 				$&$ x^7 + x^6 + x^4 + 1
 				$ \\ \hline

 				$[ 15 , 7 , 5 ]_ 2 $& $[ 15 , 4 , 8 ]_ 2 $&$ x^8 + x^4 + x^2 + x + 1
 				$&$ x^{11} + x^{10} + x^9 + x^8 + x^6 + x^4 + x^3 + 1
 				$ \\ \hline
 				
 				$[ 15 , 7 , 5 ]_ 2 $& $[ 15 , 8 , 4 ]_ 2 $&$ x^8 + x^7 + x^6 + x^4 + 1
 				$&$ x^7 + x^6 + x^4 + 1
 				$ \\ \hline

 				$[ 15 , 9 , 4 ]_ 2 $& $[ 15 , 6 , 6 ]_ 2 $&$ x^6 + x^4 + x^3 + x^2 + 1
 				$&$ x^9 + x^7 + x^6 + x^3 + x^2 + 1
 				$ \\ \hline
 				$[ 15 , 9 , 4 ]_ 2 $& $[ 15 , 2 , 10 ]_ 2 $&$ x^6 + x^4 + x^3 + x^2 + 1
 				$&$ x^{13} + x^{12} + x^{10} + x^9 + x^7 + x^6 + x^4 + x^3 + x + 1
 				$ \\ \hline

 				$[ 15 , 10 , 4 ]_ 2 $& $[ 15 , 4 , 8 ]_ 2 $&$ x^5 + x^3 + x + 1
 				$&$ x^{11} + x^{10} + x^9 + x^8 + x^6 + x^4 + x^3 + 1
 				$ \\ \hline
 				$[ 15 , 11 , 3 ]_ 2 $& $[ 15 , 4 , 8 ]_ 2 $&$ x^4 + x + 1
 				$&$ x^{11} + x^8 + x^7 + x^5 + x^3 + x^2 + x + 1
 				$ \\ \hline
 				
 				$[ 15 , 13 , 2 ]_ 2 $& $[ 15 , 2 , 10 ]_ 2 $&$ x^2 + x + 1
 				$&$ x^{13} + x^{12} + x^{10} + x^9 + x^7 + x^6 + x^4 + x^3 + x + 1
 				$ \\ \hline
 				
 				$[ 17 , 8 , 6 ]_ 2 $& $[ 17 , 8 , 6 ]_ 2 $&$ x^9 + x^6 + x^5 + x^4 + x^3 + 1
 				$&$ x^9 + x^8 + x^6 + x^3 + x + 1
 				$ \\ \hline
 				$[ 17 , 9 , 5 ]_ 2 $& $[ 17 , 8 , 6 ]_ 2 $&$ x^8 + x^7 + x^6 + x^4 + x^2 + x + 1
 				$&$ x^9 + x^8 + x^6 + x^3 + x + 1
 				$ \\ \hline

 			\end{tabular}
 			\caption{Optimal linear $0$-Intersection Pairs of Binary Cyclic Codes} \label{t1}
 		\end{table}

 		\begin{table}
 			\centering
 			\begin{tabular}{|c|c|l|l|}
 				\hline
 				$C$ &$D$ & $g_C(x)$& $g_D(x)$ \\ \hline
 				
 				$[ 7 , 4 , 3 ]_ 2 $& $[ 7 , 4 , 3 ]_ 2 $&$ x^3 + x^2 + 1
 				$&$ x^3 + x + 1
 				$ \\ \hline

 				$[ 15 , 5 , 7 ]_ 2 $& $[ 15 , 5 , 7 ]_ 2 $&$ x^{10} + x^9 + x^8 + x^6 + x^5 + x^2
 				+ 1
 				$&$ x^{10} + x^8 + x^5 + x^4 + x^2 + x + 1
 				$ \\ \hline
 				$[ 15 , 7 , 5 ]_ 2 $& $[ 15 , 5 , 7 ]_ 2 $&$ x^8 + x^4 + x^2 + x + 1
 				$&$ x^{10} + x^8 + x^5 + x^4 + x^2 + x + 1
 				$ \\ \hline
 				$[ 15 , 11 , 3 ]_ 2 $& $[ 15 , 5 , 7 ]_ 2 $&$ x^4 + x + 1
 				$&$ x^{10} + x^9 + x^8 + x^6 + x^5 + x^2 + 1
 				$ \\ \hline

 				$[ 17 , 9 , 5 ]_ 2 $& $[ 17 , 9 , 5 ]_ 2 $&$ x^8 + x^7 + x^6 + x^4 + x^2 + x + 1
 				$&$ x^8 + x^5 + x^4 + x^3 + 1
 				$ \\ \hline

 			\end{tabular}
 			\caption{Optimal linear $1$-Intersection Pairs of Binary Cyclic Codes}\label{t2}
 		\end{table}

 		\begin{table}
 			\centering
 			\begin{tabular}{|c|c|l|l|c|}
 				\hline
 				$C$ &$D$ & $g_C(x)$& $g_D(x)$  \\ \hline
 				
 				$[ 9 , 2 , 6 ]_ 2 $& $[ 9 , 2 , 6 ]_ 2 $&$ x^7 + x^6 + x^4 + x^3 + x + 1
 				$&$ x^7 + x^6 + x^4 + x^3 + x + 1
 				$ \\ \hline
 				
 				$[ 9 , 8 , 2 ]_ 2 $& $[ 9 , 2 , 6 ]_ 2 $&$ x + 1
 				$&$ x^7 + x^6 + x^4 + x^3 + x + 1
 				$ \\ \hline

 				$[ 15 , 6 , 6 ]_ 2 $& $[ 15 , 2 , 10 ]_ 2 $&$ x^9 + x^6 + x^5 + x^4 + x + 1
 				$&$ x^{13} + x^{12} + x^{10} + x^9 + x^7 + x^6 + x^4 + x^3 + x + 1
 				$ \\ \hline
 				
 				$[ 15 , 6 , 6 ]_ 2 $& $[ 15 , 6 , 6 ]_ 2 $&$ x^9 + x^8 + x^5 + x^4 + x^3 + 1
 				$&$ x^9 + x^7 + x^6 + x^3 + x^2 + 1
 				$ \\ \hline		 
 				
 				$[ 15 , 6 , 6 ]_ 2 $& $[ 15 , 10 , 4 ]_ 2 $&$ x^9 + x^8 + x^5 + x^4 + x^3 + 1
 				$&$ x^5 + x^4 + x^2 + 1
 				$ \\ \hline
 				
 				$[ 15 , 7 , 5 ]_ 2 $& $[ 15 , 2 , 10 ]_ 2 $&$ x^8 + x^4 + x^2 + x + 1
 				$&$ x^{13} + x^{12} + x^{10} + x^9 + x^7 + x^6 + x^4 + x^3 + x + 1
 				$ \\ \hline
 				
 				$[ 15 , 7 , 5 ]_ 2 $& $[ 15 , 6 , 6 ]_ 2 $&$ x^8 + x^4 + x^2 + x + 1
 				$&$ x^9 + x^6 + x^5 + x^4 + x + 1
 				$ \\ \hline
 				
 				$[ 15 , 7 , 5 ]_ 2 $& $[ 15 , 10 , 4 ]_ 2 $&$ x^8 + x^4 + x^2 + x + 1
 				$&$ x^5 + x^4 + x^2 + 1
 				$ \\ \hline

 				$[ 15 , 10 , 4 ]_ 2 $& $[ 15 , 2 , 10 ]_ 2 $&$ x^5 + x^3 + x + 1
 				$&$ x^{13} + x^{12} + x^{10} + x^9 + x^7 + x^6 + x^4 + x^3 + x + 1
 				$ \\ \hline	
 				
 				$[ 15 , 10 , 4 ]_ 2 $& $[ 15 , 6 , 6 ]_ 2 $&$ x^5 + x^3 + x + 1
 				$&$ x^9 + x^6 + x^5 + x^4 + x + 1
 				$ \\ \hline
 				
 				$[ 15 , 11 , 3 ]_ 2 $& $[ 15 , 2 , 10 ]_ 2 $&$ x^4 + x^3 + 1
 				$&$ x^{13} + x^{12} + x^{10} + x^9 + x^7 + x^6 + x^4 + x^3 + x + 1
 				$ \\ \hline
 				
 				$[ 15 , 11 , 3 ]_ 2 $& $[ 15 , 6 , 6 ]_ 2 $&$ x^4 + x^3 + 1
 				$&$ x^9 + x^6 + x^5 + x^4 + x + 1
 				$ \\ \hline

 				$[ 15 , 14 , 2 ]_ 2 $& $[ 15 , 2 , 10 ]_ 2 $&$ x + 1
 				$&$ x^{13} + x^{12} + x^{10} + x^9 + x^7 + x^6 + x^4 + x^3 + x + 1
 				$ \\ \hline

 			\end{tabular}
 			\caption{Optimal linear $2$-Intersection Pairs of Binary Cyclic Codes} \label{t3}
 		\end{table}

 	\end{landscape}
 	
 	\section{Conclusion and Remarks}
 	
 	Linear intersection pairs of  cyclic  codes over finite fields have been studied.  Some properties of cyclotomic cosets in cyclic groups have been  presented and recalled as key tools in the study of   such linear intersection pairs.  Characterization for the existence of a linear intersection pair of cyclic codes and constructions of  cyclic  codes pairs  of a fixed intersecting dimension    have been investigated and presented  in terms of their generator polynomials and  cyclotomic cosets.  
 	Based on the theoretical results,  some numerical  examples of linear intersection pairs of  cyclic  codes with  good parameters have been illustrated.

 	\section*{Acknowledgments}
 	This project is funded by the National Research Council of Thailand and Silpakorn University  under Research Grant  N42A650381.


\begin{thebibliography}{99}
 		\bibitem{BCP1997}  W. Bosma, J. Cannon, C.  Playoust, The Magma algebra system. I. The user language, {\em J. Symbolic Comput.},  {\bf 24}  (1997),  235--265. \url{https://doi.org/10.1006/jsco.1996.0125}
 		
 		
 		
 		
 		
 		
 		\bibitem{CG2016} C. Carlet, S. Guilley, Complementary dual codes for counter-measures to side-channel attacks,  {\em Adv. Math. Commun.},   {\bf 10} (2016),  131--150. \url{https://doi/10.3934/amc.2016.10.131}
 		
 		
 		\bibitem{carlet1} C. Carlet,  C. G\"uneri, S. Mesnager,  F. \"Ozbudak,
 		Construction of some codes suitable for both side channel and fault injection attacks, {\em Proceedings of International Workshop on the Arithmetic of Finite Fields (WAIFI 2018)}, Bergen, 2018. 
 		
 		
 		
 		
 		
 		
 		\bibitem{carlet2} C. Carlet, C. G\"uneri,  F. \"Ozbudak, B. \"Ozkaya, P. Sol\'e,  On linear complementary pairs of codes, {\em IEEE Trans. Inform. Theory},  {\bf 64} (2018), 6583--6589. \url{https://doi.org/10.1109/TIT.2018.2796125}
 		
 		
 		
 		%	\bibitem{CMTQ2018}
 		%	Carlet, C.,
 		%	Mesnager, S.,
 		%	Tang, C., Qi, Y.: Euclidean and Hermitian LCD MDS codes, Des. Codes
 		%	Cryptogr.  {\bf  86}, 2605--2618 (2018).
 		
 		%	\bibitem{CG2015}  X.	Chang, J.  Gao,   Self-orthogonal cyclic codes and complementary-dual codes of length $p^nq^m$, British Journal of Mathematics \& Computer Science, {\bf 8}, 401--410 (2015).
 		%	
 		%	
 		
 		\bibitem{G2023}   M. Grassl,  Bounds on the minimum distance of linear codes and quantum codes,  
 		Online available at http://www.codetables.de.
 		Accessed on 2023-09-02.
 		
 		\bibitem{GGJT2020}  K.  Guenda, T. A. Gulliver,  S.  Jitman, S.Thipworawimon, Linear $\ell$-intersection pairs of codes and their applications, {\em Des. Codes Cryptogr.},  {\bf 88}  (2020),  133--152. \url{https://doi.org/10.1007/s10623-019-00676-z}
 		
 		
 		
 		
 		
 		\bibitem{Huff03} W. C. Huffman,  V. Pless, {\em Fundamentals of error-correcting codes}, Cambridge: Cambridge University Press, 2003.  
 		
 		\bibitem{Jia11} Y. Jia,  S. Ling, C. Xing, On self-dual cyclic codes over finite fields, {\em  IEEE Trans. Inform. Theory}, {\bf 57} (2011),  2243--2251.  \url{https://doi.org/10.1109/TIT.2010.2092415}
 		
 		\bibitem{Jin2017}  L. Jin,  Construction of MDS codes with complementary duals,   {\em IEEE Trans. Inform. Theory},  {\bf 63} (2017), 2843--2847. \url{https://doi.org/10.1109/TIT.2016.2644660}
 		
 		\bibitem{JLLX2012} S. Jitman, S. Ling, H. Liu,  H., X.   Xie,  Abelian codes in principal ideal group algebras, {\em IEEE Trans. Info. Theory},   {\bf 59} (2013),  3046--3058.  \url{https://doi.org/10.1109/TIT.2012.2236383}
 		
 		
 		\bibitem{KR2012}	L. Kathuria, M. Raka,  Existence of cyclic self-orthogonal codes: A note on a result of Vera Pless, {\em Advances in Mathematics of Communications},  {\bf 6} (2012), 499--503.  \url{https://doi.org/10.3934/amc.2012.6.499}
 		
 		\bibitem{Lidl97} R. Lidl, H. Niederreiter,  {\em  Finite fields}, Cambridge: Cambridge University Press,   1997.
 		
 		
 		
 		
 		
 		\bibitem{Mac} F. J. MacWilliams,  N. J. A.  Sloane, {\em   The theory of error-correcting codes}, Amesterdam: North-Holland Publishing,  1977.
 		
 		\bibitem{M1992} J.L. Massey, 
 		Linear codes with complementary duals,
 		{\em Discrete Math.},  {\bf 106-107}  (1992),
 		337--342. \url{https://doi.org/10.1016/0012-365X(92)90563-U}
 		
 		\bibitem{QZ2018}  J. Qian, L. Zhang, 
 		On MDS linear complementary dual codes and entanglement-assisted quantum codes, {\em Des. Codes Cryptogr.},  {\bf  86} (2018),   1565--1572. \url{https://doi.org/10.1007/s10623-017-0413-x}
 		
 		
 		
 		
 		\bibitem{SJU2015} E.	Sangwisut, S. Jitman, S. Ling, P. Udomkavanich,  Hulls of cyclic and negacyclic codes over finite fields, {\em Finite Fields Appl.}, {\bf 33} (2015), 232--257. \url{https://doi.org/10.1016/j.ffa.2014.12.008}
 		
 		\bibitem{S2003} 	G. Skersys,  The average dimension of the hull of cyclic codes,  {\em Discrete Appl. Math.},  {\bf 128} (2003),  275--292. \url{https://doi.org/10.1016/S0166-218X(02)00451-1}
 		
 		
 		\bibitem{TJ2020}  S. Thipworawimon, S. Jitman,  Hulls of linear codes revisited with applications, {\em  J. Appl. Math. Comput.},  {\bf 62} (2020), 325--340.  \url{https://doi.org/10.1007/s12190-019-01286-7}
 		
 		\bibitem{YM1998} X. Yang,  J. L. Massey, 
 		The condition for a cyclic code to have a complementary dual,
 		{\em Discrete  Math.},
 		{\bf 126} 	(1994),  391--393. \url{https://doi.org/10.1016/0012-365X(94)90283-6}
 		
 	\end{thebibliography}
\end{document}